\documentclass[conference]{IEEEtran}
\IEEEoverridecommandlockouts
\usepackage{amsmath,amssymb,amsfonts, amsthm}
\usepackage{algorithmic}
\usepackage{graphicx}
\usepackage{textcomp}
\usepackage{xcolor}
\usepackage{subcaption}
\usepackage{microtype}
\usepackage{multicol}
\usepackage{bm}
\usepackage{booktabs}
\usepackage{cite}
\newtheorem{theorem}{Theorem}

\def\BibTeX{{\rm B\kern-.05em{\sc i\kern-.025em b}\kern-.08em
    T\kern-.1667em\lower.7ex\hbox{E}\kern-.125emX}}
\begin{document}

\title{MDM: Manhattan Distance Mapping of DNN Weights for Parasitic-Resistance‑Resilient Memristive Crossbars}

\author{\IEEEauthorblockN{Matheus Farias}
\IEEEauthorblockA{\textit{Harvard University} \\
matheusfarias@g.harvard.edu}
\and
\IEEEauthorblockN{Wanghley Martins}
\IEEEauthorblockA{
\textit{Duke University}\\
wanghley.soares.martins@duke.edu}
\and
\IEEEauthorblockN{H. T. Kung}
\IEEEauthorblockA{
\textit{Harvard University}\\
kung@harvard.edu}
}

\maketitle

\begin{abstract}
\textit{\underline{M}anhattan \underline{D}istance \underline{M}apping} (MDM) is a post-training deep neural network (DNN) weight mapping technique for memristive bit-sliced compute-in-memory (CIM) crossbars that reduces parasitic resistance (PR) nonidealities.

PR limits crossbar efficiency by mapping DNN matrices into small crossbar tiles, reducing CIM-based speedup. Each crossbar executes one tile, requiring digital synchronization before the next layer. At this granularity, designers either deploy many small crossbars in parallel or reuse a few sequentially—both increasing analog-to-digital conversions, latency, I/O pressure, and chip area.

MDM alleviates PR effects by optimizing active-memristor placement. Exploiting bit-level structured sparsity, it feeds activations from the denser low-order side and reorders rows according to the Manhattan distance, relocating active cells toward regions less affected by PR and thus lowering the nonideality factor (NF).

Applied to DNN models on ImageNet-1k, MDM reduces NF by up to 46\% and improves accuracy under analog distortion by an average of 3.6\% in ResNets. Overall, it provides a lightweight, spatially informed method for scaling CIM DNN accelerators.
\end{abstract}


\section{Introduction}
Compute-in-memory (CIM) architectures integrate storage and computation within the same physical fabric, offering energy-efficient deep neural network (DNN) acceleration by reducing data movements \cite{jhang2021, yu2021, ali2022, kaur2024}. However, their scalability remains limited by nonidealities—such as sneak paths \cite{lee2025, rao2022, cassuto2013, cassuto2016}, process-voltage-temperature variations \cite{chen2025, jo2010}, stuck-at faults \cite{you2025, oli2022, zhang2019, yeo2019}, conductance drift \cite{munoz2021, ambrogio2019}, and parasitic resistance (PR) \cite{xu2022, zhang2021, zhang2020}—which degrade inference accuracy and restrict computational parallelism in large-scale workloads \cite{rasch2023, ciprut2017}.

PR is a key scalability bottleneck in CIM accelerators, caused by resistive interconnects within crossbars.
We hypothesize that the resulting voltage drops grow proportionally with the Manhattan distance from the I/O rails—an effect we term \textit{the Manhattan Hypothesis}. This model enables analytical estimation of PR impact without requiring circuit-level simulations.

DNN weights are typically mapped across crossbar rows, with each column representing a fractional bit \cite{farias2025a, farias2025b, shafiee2016, chou2019}.
Because weights follow a bell-shaped distribution centered near zero \cite{farias2025a, farias2025b, fang2020, horton2022, tambe2020}, high-order columns that encode large magnitudes are sparse, while lower-order columns are more frequently active.
This structured imbalance drives current through deeper paths, leading to amplified PR effects.

This spatial nonideality constrains crossbar size. Large arrays amplify PR deviations, forcing DNN partitioning into smaller crossbar tiles to preserve accuracy. However, smaller tiles demand additional digital synchronization and I/O bandwidth between computing phases, mitigating CIM intra-parallelism. Consequently, PR simultaneously degrades model accuracy and undermines system-level throughput, posing a fundamental obstacle to scaling CIM-based DNN accelerators.

To address this limitation, we propose the \textit{Manhattan Distance Mapping} (MDM) algorithm, a post-training spatial remapping strategy that reduces PR distortion without altering crossbar computation.
MDM operates in three stages.
First, it reverses the dataflow so that denser, lower-order bit regions—where active memristors are concentrated—align with shorter conduction paths, reducing cumulative voltage drops.
Second, it assigns each row a Manhattan-based score that quantifies the distance of its active cells from the I/O rails, reflecting their relative exposure to parasitic effects.
Finally, rows are reordered in ascending order of this score, relocating dense regions toward areas less affected by resistance buildup.
This spatial reorganization reduces the nonideality factor (NF)—the deviation of the measured output from its ideal value—while preserving all arithmetic semantics, requiring neither retraining nor hardware modification, and integrating seamlessly into existing deployments \cite{farias2025a, farias2025b}. See Figure \ref{fig:summary} for a summary of the approach.

The main contributions of this paper are:
\begin{itemize}
\item A theoretical foundation for MDM, built upon (1) the \textit{Manhattan Hypothesis}, which shows that voltage drops accumulate proportionally to the Manhattan distance from the I/O rails, and (2) a mathematical proof of structured bit-level sparsity in DNN weight distributions;
\item A post-training weight mapping that reverses dataflow and reorders rows to place active cells in regions less affected by PR accumulation, requiring no retraining nor hardware modification;
\item A framework that models PR by injecting spatially dependent noise into DNN weights, enabling analog distortion assessment on PyTorch models.
\end{itemize}

\begin{figure*}
    \centering
    \includegraphics[width=\textwidth]{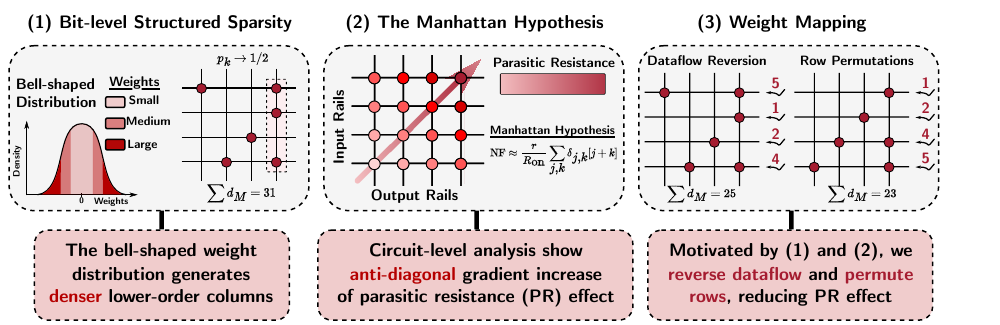}
    \caption{Summary of the Manhattan Distance Mapping (MDM).}
    \label{fig:summary}
\end{figure*}

\section{Background}
\label{sec:back}
\subsection{Memristive Crossbar Assumptions}
\label{sec:ass}

We adopt bit-sliced crossbars \cite{farias2025a, farias2025b}, where each row encodes a weight and columns represent power-of-two scaling factors. Higher-order columns near the inputs correspond to larger factors (e.g. $2^0, 2^{-1}, 2^{-2}$), while lower-order columns farther away encode smaller ones (e.g. $2^{-5}, 2^{-6}, 2^{-7}$). For a 128×128 crossbar with 16 multipliers, each row stores eight different weight values (since $128/16=8$).

This hierarchy produces structured sparsity according to the bell-shaped distribution of DNN weights (see Section \ref{sec:sparse}).

\subsection{Nonideality Measurement}
\label{sec:nf}
Crossbar nonidealities are quantified by the NF \cite{bhatta2022, chakra2020, chakraborty2020},
\begin{equation}
    \label{eq:definition_noni}
    \text{NF}=\left|{\frac{\Delta i}{i_0}}\right |,
\end{equation}
where $i_0$ is the expected output and $\Delta i$ is the amount of current that was deviated due to nonidealities.

Each cell can be identified by its position ($j,k$) corresponding to its row and column indices seeing from the I/O interface. The Manhattan distance, $d_{\text{M}}(j,k)$, of a cell is defined as the sum of its horizontal and vertical distances from the I/O rails, 
\begin{equation}
    d_{\text{M}}(j,k) = j+k.   
\end{equation}
As current propagates along the resistive mesh, voltage drops accumulate with increasing distance from the I/O rails, causing farther cells to contribute less accurately to the overall output. Section \ref{sec:manhattan} uses Kirchhoff’s law to hypothesize that the NF grows proportionally to the Manhattan distance of active cells.

\section{Theoretical Framework}
\label{sec:theory}
This section proves that (1) bit-sliced crossbars exhibits a structured bit-level sparsity pattern, and propose that (2) the NF scales with the Manhattan distance from the I/O rails—two properties that underpin the MDM method.

\subsection{Bit-level Structured Sparsity}
\label{sec:sparse}
To characterize the bit-level distribution in a $(J,K)$ crossbar, we apply Theorem \ref{th:sparsity}, derived from the DNN bell-shaped weight distribution \cite{farias2025a, farias2025b, han2015, fang2020, horton2022, tambe2020}.  
Each weight $w_j$ is mapped across $K$ fractional-bit columns as $w_j = \sum_{k \leq K} b_{j,k}(w_j)2^{-k},$
where lower-order bits exhibit higher activation probability $(b_k = 1)$\footnote{We supress the row index when the statement is row-independent.}, yielding denser columns.

\begin{theorem}
\label{th:sparsity}
    Let $W$ be a nonnegative random variable with probability density function $f : [0, \infty[ \to [0, \infty[$ such that:
    \begin{enumerate}
        \item $f$ is continuous on $[0, \infty[$ and strictly decreasing on $]0, \infty[$;
        \item $f(0) < \infty$ and $\lim_{w \to \infty} f(w) = 0$.
    \end{enumerate}
Let
\[
p_k := \mathbb{P}(b_k = 1) = \int_0^{\infty} f(w) b_k(w) dw,
\]
where $\mathbb{P}(b_k = 1)$ is the probability of $b_k = 1$.
Then
\[
\boxed{
\bigl|p_k - \tfrac{1}{2}\bigr|
    \le \frac{1}{2^{2+k}} f(0).
}
\]
In particular, $p_k < 1/2$ for every $k$ and $p_k \to 1/2$ as $k \to \infty$.
\end{theorem}

\begin{proof}
For $k \leq K$, set $L := 2^{-k}$ and define the $k$-th fractional-bit indicator
\[b_k(w) =
\begin{cases}
0, & w \in [mL,\, mL + \tfrac{L}{2}[, \\
1, & w \in [mL + \tfrac{L}{2},\, (m+1)L[,
\end{cases}
\quad m = 0, 1, 2, \dots   
\]
    Let
    \begin{equation}
        \Delta_k := \mathbb{P}(b_k = 0) - \mathbb{P}(b_k = 1).
    \end{equation}
    Then,
    \begin{equation}
        \Delta_k = \sum_{m=0}^{\infty} \int_{0}^{L/2} 
      \big[ f(mL + u) - f(mL + u + L/2) \big]\, du.
    \end{equation}
    By the Fundamental Theorem of Calculus,
\begin{equation}
    f(mL + u) - f(mL + L/2 + u)
   = - \int_{0}^{L/2} f'(mL + u + \theta)\, d\theta.
\end{equation}   
   Changing variables to  $s = u + \theta \in [0, L] $ gives
\begin{equation}
\Delta_k
   = \sum_{m=0}^{\infty} \int_{0}^{L}
       \big[ -f'(mL + s) \big]\, A_L(s)\, ds, 
\end{equation}
with \(A_L(s) := \min\{s, L - s\}\) on \([0, L]\).
Since $A_L(s) \le L/2$,
\begin{equation}
\Delta_k
   \le \frac{L}{2} \sum_{m=0}^{\infty}
       \int_{0}^{L} \big[ -f'(mL + s) \big]\, ds    
\end{equation}
using the Fundamental Theorem of Calculus again
\begin{equation}
  \Delta_k \le \frac{L}{2} \sum_{m=0}^{\infty}
       \big[ f(mL) - f((m+1)L) \big].  
\end{equation}
The series telescopes and $f((m+1)L) \to 0$ (statement 2 of the theorem), hence
\begin{equation}
\Delta_k \le \frac{L}{2} f(0).
\end{equation}
Noting that $p_k = \tfrac{1}{2}(1 - \Delta_k)$ and $L=2^{-k}$, we obtain
\begin{equation}
\big|p_k - \tfrac{1}{2}\big|
   \le \frac{1}{2^{2+k}} f(0).    
\end{equation}
Moreover, $\Delta_k > 0$ because $f$ is strictly decreasing on sets of positive measure (statement 1 of the theorem),
hence $p_k < \tfrac{1}{2}$.
Finally, since $L = 2^{-k} \to 0$, the bound forces $p_k \to \tfrac{1}{2}$.
\qedhere
\end{proof}
Theorem \ref{th:sparsity} motivates dataflow reversion by injecting inputs from the denser side to minimize PR along conductive paths.
\subsection{The Manhattan Hypothesis}
\label{sec:manhattan}
We consider a $(J,K)$ crossbar whose interconnects have parasitic resistance $r$, and where each active cell at $(j,k)$ exhibits resistance $R_\mathrm{on}$. The array is driven from $V_\mathrm{in}$ along the rows and sensed at the grounded column outputs.

A single active memristor $\ell$ cells farther from the input rail satisfies the Kirchhoff’s law:
\begin{equation}
\frac{V - V_{\mathrm{in}}}{\ell r} + \frac{V}{R_{\mathrm{on}}} = 0.
\label{eq:kcl_wire}
\end{equation}
The memristor current, including PR effects, is
\begin{equation}
i = i_0 + \Delta i = \frac{V}{R_{\mathrm{on}}},
\label{eq:i_actual}
\end{equation}
where $i_0=V_0/R_{\mathrm{on}}$ is the ideal current (for $r=0$).

Solving Equation \eqref{eq:kcl_wire} for $V$ under the practical assumption $\ell r \ll R_\mathrm{on}$ gives the first-order approximation
\begin{equation}
V\approx V_0 \left[ 1 - \frac{\ell r}{R_{\mathrm{on}}} \right].
\label{eq:linearized_V}
\end{equation}
Substituting Equation \eqref{eq:linearized_V} into \eqref{eq:i_actual} and normalizing by $i_0$:
\begin{equation}
\mathrm{NF}
= \left| \frac{\Delta i}{i_0} \right|
= \frac{|V - V_0|}{V_0}
\approx \ell \frac{r}{R_{\mathrm{on}}}.
\label{eq:nf_single}
\end{equation}
Equation~\eqref{eq:nf_single} shows that the deviation increases linearly with the distance between the device and its I/O rail. For an active memristor located $j$ segments from the input rail and $k$ segments from the output rail, the combined contribution is
\begin{equation}
\boxed{
\mathrm{NF} \approx \frac{r}{R_{\mathrm{on}}}[j + k]
}.
\label{eq:nf_total}
\end{equation}
Extending for multiple active memristors and sensing the current at each column end, we obtain the Manhattan Hypothesis
\begin{equation}
    \boxed{
    \mathrm{NF} \approx \frac{r}{R_\mathrm{on}}\sum_{j,k}\delta_{j,k}[j+k]
    }
    \quad \text{(Manhattan Hypothesis)}
    \label{eq:nf_column}
\end{equation}
where $\delta_{j,k} = 1$ if the crosspoint $(j,k)$ is active and $0$ otherwise.

This result shows that the NF scales proportionally with the aggregate Manhattan distance of active cells, following a gradient of increase from the bottom-left to the top-right (anti-diagonal) of the array. Consequently, crossbars exhibit identical NF values under anti-diagonal symmetric configurations—a behavior corroborated by SPICE circuit-level simulations (see Figure~\ref{fig:symmetry}).
This linear relationship isolates PR as the sole source of nonideality. Other effects, such as sneak-path currents, are not captured by this first-order model.
To decouple these phenomena, we consider the sparse regime of bit-sliced crossbars for DNN workloads. In such configurations, sneak paths are more likely to be suppressed \cite{cassuto2013, cassuto2016}.
\begin{figure}
    \centering
\includegraphics[width=\columnwidth]{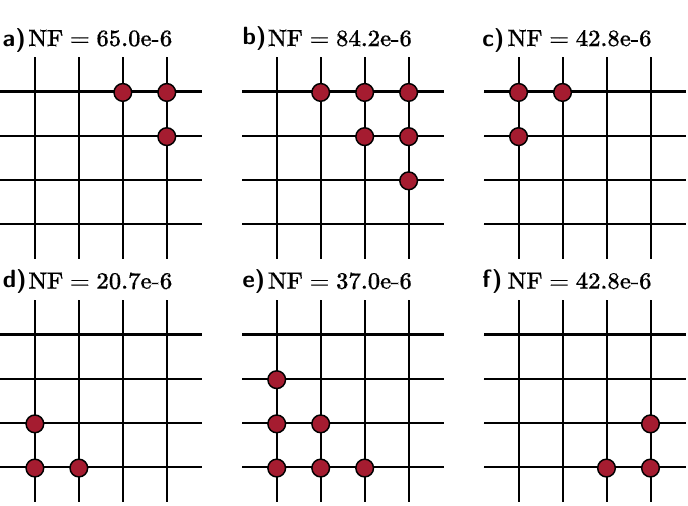}
    \caption{Circuit-level simulations in SPICE shows anti-diagonal symmetry for $r = 2.5$ $\Omega$, $R_{\text{on}} = 300$ k$\Omega$, and $R_{\text{off}} = 3$ M$\Omega$ (values within range suggested in the literature \cite{chakraborty2020, cao2025, cao2021}.)}
    \label{fig:symmetry}
\end{figure}

\section{Manhattan Distance Weight Mapping}

The MDM algorithm reduces crossbar PR effect by reorganizing weights to minimize the Manhattan distance of active memristors from the I/O rails in three steps.

First, the dataflow is reversed so that denser, lower-order bits align with shorter conduction paths, thereby reducing the PR impact.  
Second, a Manhattan-based score is computed for each row, quantifying the distance of its active memristors from the input.  
Finally, rows are sorted according to this score, positioning denser rows closer to I/O (see Figure \ref{fig:mdwm}).  
\begin{figure}
    \centering
    \includegraphics[width=\columnwidth]{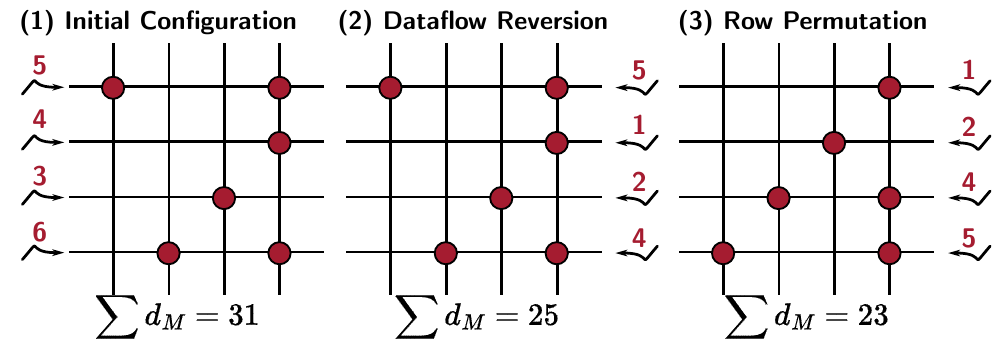}
    \caption{MDM example. Arrows on the left/right indicate dataflow and numbers on top of each arrow indicate row score.}
    \label{fig:mdwm}
\end{figure}

This spatial remapping minimizes NF without modifying the DNN model. It operates as a post-training transformation that can be seamlessly integrated into existing CIM deployments.

From a system-level perspective, row permutations and reversed dataflow require buffer drivers and multiplexing circuitry already present in state-of-the-art CIM implementations \cite{farias2025a,farias2025b}. The approach extends these architectures by modeling and reducing crossbar nonidealities with a novel mapping policy.

\section{Experiments}
\label{sec:exp}
We assess (1) how accurate the Manhattan Hypothesis is and (2) NF reduction and (3) model accuracy drop considering PR effects before/after MDM, benchmarking multiple DNNs. Crossbar computations were simulated in SPICE and PyTorch on ImageNet-1K \cite{deng2009} on all model layers (ResNets, VGGs, ViTs and DeITs from native PyTorch models), trained in 32-bit floating point. The simulations used 128x10 crossbars in 64x64 tiles with the same resistance values as in Section \ref{sec:manhattan}.
\subsection{The Manhattan Hypothesis Accuracy}
\label{sec:spice}
We evaluate the Manhattan Hypothesis in three stages: (1) we generate 500 randomized crossbar tiles with approximately 80\% sparsity, matching the lower bound observed across the evaluated models. Since the least sparse model, DeiT-Base, exhibits 76\% sparsity, this level ensures consistency with all architectures, whose sparsity is at least 80\%; (2) each tile is simulated in SPICE. The circuit-level simulation measures the NF by probing the column outputs for $r = 0$ (expected output) and $r = 2.5$ $\Omega$ (actual output affected by PR); (3) we apply least-squares to find the linear map between the measured and calculated NF\footnote{We calculate NF from Equation \eqref{eq:nf_column} and measure it using SPICE.} (see Figure \ref{fig:distribution}).

\begin{figure}
    \centering
\includegraphics[width=\columnwidth]{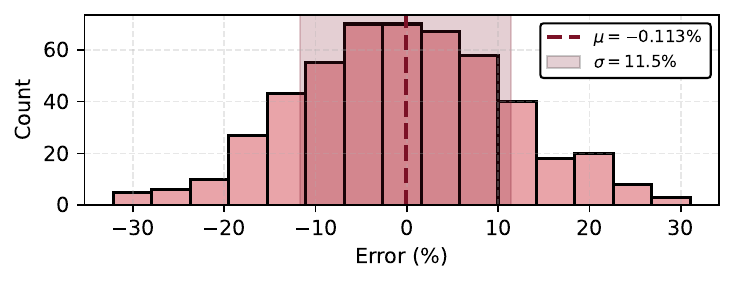}
    \caption{Error distribution of the Manhattan Hypothesis linear fit has mean $\mu = -0.126\%$ and standard deviation $\sigma = 11.2\%$.}
    \label{fig:distribution}
\end{figure}
\subsection{Nonideality Factor Reduction}
The Manhattan hypothesis allows fast PyTorch NF evaluation without exhaustive circuit-level simulation of every DNN tile. As illustrated in Figure \ref{fig:nf}, MDM significantly reduces the NF. By comparing dataflows, we observe that reverted dataflow improves MDM by up to 50\% compared to conventional.
\begin{figure}
    \centering
\includegraphics[width=\columnwidth]{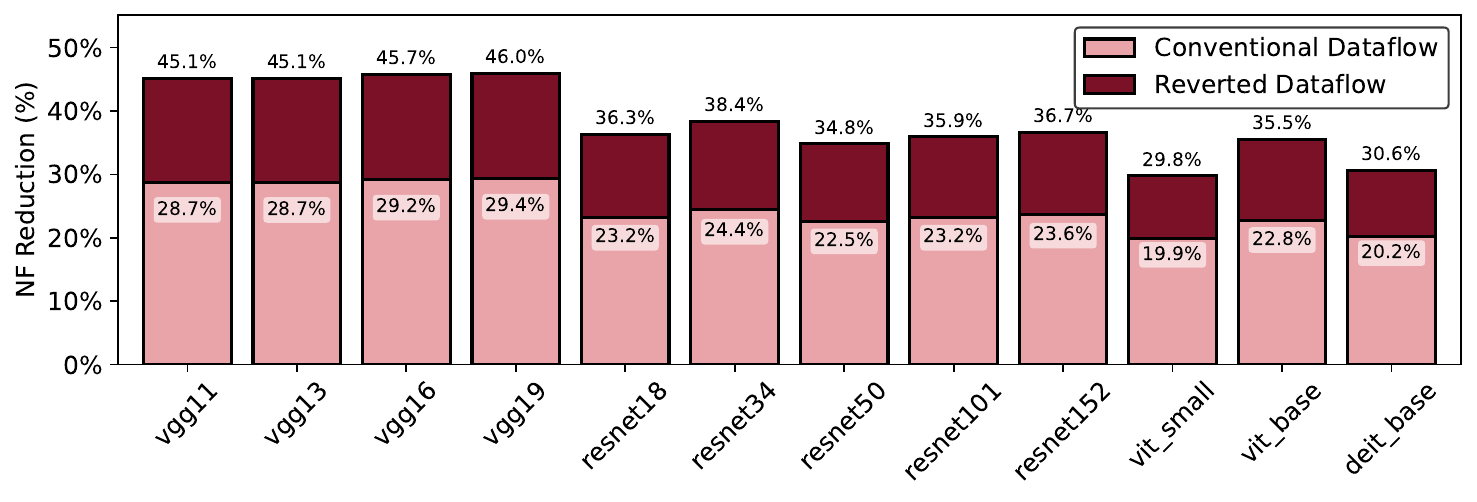}
    \caption{NF reduction with MDM for different dataflows.}
    \label{fig:nf}
\end{figure}
\subsection{Model Accuracy Evaluation}
Finally, we translate the NF reduction to model accuracy by injecting position-dependent noise in PyTorch, where each weight is modified proportionally to the Manhattan distance:
\begin{equation}
\label{eq:translate}
    w'_j=\sum_{k \leq K}b_{j,k}(w_j)2^{-k}[1 + \eta \delta_{j,k}],
\end{equation}
where $\eta$ is the noise coefficient.

The parameter $\eta$ is calibrated in SPICE using Equation \eqref{eq:translate}, such that simulations with $r=2.5$ $\Omega$ match the ideal $r=0$ case. This procedure yields $\eta = 2\times 10^{-3}$. Figure \ref{fig:acc} reports model accuracy under noise injection with and without MDM.
\begin{figure}
    \centering
\includegraphics[width=\columnwidth]{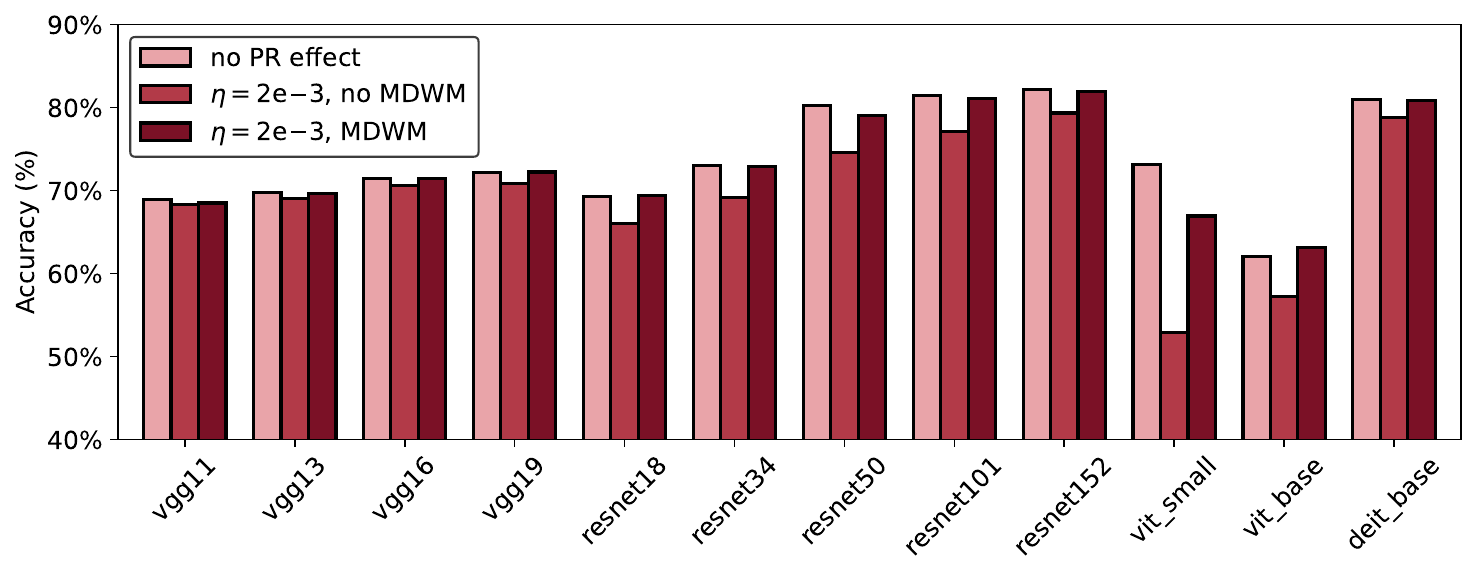}
    \caption{Model accuracy for different configurations.}
    \label{fig:acc}
\end{figure}
Overall, MDM tends to be less effective for transformer models due to their characteristically flatter weight distributions \cite{farias2025a, farias2025b, tambe2020, bonda2021}. As a result, their bit-line representations become denser in higher-order columns and sparser in lower-order ones, which diminishes the benefits of MDM.
\section{Conclusion}
We introduced the \textit{Manhattan Distance Mapping} (MDM), a spatially informed post-training weight mapping method that reduces parasitic resistance effects in memristive compute-in-memory (CIM) crossbars. By reversing the dataflow and reordering rows according to their cumulative Manhattan distance from the I/O rails, MDM effectively relocates active memristors toward regions less affected by PR voltage drops. The method considerably reduces the nonideality factor (NF).

Through circuit-level and PyTorch-based simulations on ImageNet-1k, we demonstrated that MDM decreases NF by up to 46\% and improves inference accuracy under analog distortion by an average of 3.6\% in ResNet architectures. These results enable larger crossbars to operate with reduced PR degradation.

By bridging algorithmic and device-level constraints, MDM opens new directions for understanding nonidealities in CIM.
\bibliographystyle{ieeetr}
\bibliography{refs}
\end{document}